\documentclass{llncs}

\usepackage{amsmath}
\usepackage{amssymb}
\usepackage{epsfig}
\usepackage{tikz,pgfplots}    
\usetikzlibrary{automata,petri,positioning,shapes,calc,
	decorations.pathmorphing,patterns,arrows,fit,backgrounds,
	shapes.multipart,fadings,shapes.geometric}

\title{A Uniform Framework for Problems on Context-Free Grammars}

\author{Javier Esparza \and Peter Rossmanith \and Stefan Schwoon}
\institute{
  Technische Universit\"at M\"unchen,
  Arcisstr. 21, 80290 M\"unchen, Germany\\
  \{\texttt{esparza,rossmani,schwoon}\}\texttt{@in.tum.de}
 } 

\def    \e      {\varepsilon}

\def	\gets	{\leftarrow}

\def	\qf	{q_{\it f}}

\def	\pre	{{\it pre}}
\def	\prestar{{\pre^*}}

\def	\To	{\Rightarrow}
\def	\Tostar	{\mathrel{\stackrel{\smash{*}}\To}}

\def	\longhookrightarrow{\lhook\joinrel\relbar\joinrel\rightarrow}

\def\Rulesh#1#2#3{\langle#1\rangle\;{\buildrel{\smash{\lower2pt\hbox{$\scriptstyle{#2}$}}}\over\longhookrightarrow}\
		\langle#3\rangle}

\def	\overrel#1#2{\mathrel{\mathop{\kern0pt #2}\limits^{#1}}}
\def	\underrel#1#2{\mathrel{\mathop{\kern0pt #1}\limits_{#2}}}
\def	\bothrel#1#2#3{\mathrel{\mathop{\kern0pt #2}\limits^{#1}_{#3}}}

\def\newarrow#1{\mathop{{\hbox{\setbox0=\hbox{$\scriptstyle{#1\quad}$}{$\buildrel{\>#1\>}\over{\hbox to \wd0{\rightarrowfill}}$}}}}}
\def\snewarrow#1{\mathop{{\hbox{\setbox0=\hbox{$\scriptstyle{#1\quad}$}{$\buildrel{\smash{\>#1\>}}\over{\hbox to \wd0{\rightarrowfill}}$}}}}}
\def\fnewarrow#1{\mathop{{\hbox{\setbox0=\hbox{$\scriptscriptstyle{#1\quad}$}{$\buildrel{\smash{\>\scriptscriptstyle{#1}\>}}\over{\hbox to \wd0{\rightarrowfill}}$}}}}}

\def\smallarrow#1{\mathop{{\hbox{\setbox0=\hbox{$\scriptstyle{#1\quad}$}{$\buildrel{\,#1\>}\over{\hbox to \wd0{\rightarrowfill}}$}}}}}

\def\Newarrow#1#2{%
\setbox0=\hbox{$\scriptstyle{#1\quad}$}%
\mathrel{\smash{\mathop{\hbox to \wd0{\rightarrowfill}}\limits^{\smash{\lower1pt\hbox{$\scriptstyle{#1}$}}}_{\smash{\raise2pt\hbox{$\scriptstyle{#2}$}}}}}%
}

\def\pmb#1{\setbox0=\hbox{#1}%
        \kern-0.05em\copy0\kern-\wd0
        \kern.05em\copy0\kern-\wd0
        \kern-.025em\raise.0433em\box0}

\newcounter{zeile}
\newbox\kasten
\setbox\kasten=\hbox{00}
\let\graph=\par
{\catcode`\[=\active\catcode`\>=\active
\gdef	\addto#1#2{$#1\leftarrow#1\cup\{#2\}$}
\gdef	\forall#1{[for all] $#1$ [do]}
\gdef	\while#1{[while] $#1$ [do]}
\gdef	\ifthen#1{[if] $#1$ [then]}

\gdef\algo{\catcode`\~=\active
	\catcode`\[=\active\catcode`\>=\active\setcounter{zeile}{0}
	\def\par{\refstepcounter{zeile}\graph\noindent\kern\wd\kasten%
		\llap{{\small\thezeile}}\quad}
	\def[##1]{{\bf##1}}\def~##1~{\mathchar"405B##1\mathchar"505D}
	\def>{\quad}\obeylines}}

\begin{document}

\maketitle

\begin{abstract}
In~\cite{manyauthors}, Bouajjani and others presented an automata-based
approach to a number of elementary problems on context-free grammars.
This approach is of pedagogical interest since it provides a uniform
solution to decision procedures usually solved by independent algorithms in
textbooks, e.g.~\cite{HU79}. This paper improves upon~\cite{manyauthors}
in a number of ways. We present a new algorithm which not only has a better
space complexity but is also (in our opinion) easier to read and understand.
Moreover, a closer inspection reveals that the new algorithm is
competitive to well-known solutions for most (but not all)
standard problems.
\end{abstract}

\section{Introduction}
\label{s:intro}

Textbooks on formal languages and automata (e.g.~\cite{HU79}) discuss
solutions for certain standard problems on context-free grammars like
tests for membership, emptiness, and finiteness. For many of these
problems independent solutions are given, e.g.\ reductions to a
graph-theoretic problem for finiteness or the well-known algorithm by
Cocke, Younger and Kasami~\cite{kasami65,younger67} for the membership
problem. Closer inspection reveals that several problems reduce to
reachability questions between sentential forms.

It has been observed that finite automata can play a useful role in
solving reachability questions. Book and Otto, in their work
about string rewriting systems~\cite{bookotto}, established a connection to
context-free grammars. More to the point, they discovered that for
a regular set~$L$ of strings over variables and terminals, the
predecessors of~$L$ -- i.e.\ the strings from which some element of~$L$
can be derived by repeatedly applying the productions -- also form a
regular set. Very similar results were discovered independently by
B\"uchi~\cite{buechi64} and Caucal~\cite{caucal92}.
Book and Otto also pointed out
the applicability of their algorithm to problems on context-free
grammars without going into details.

Esparza and Rossmanith~\cite{ER97}, attracted by this idea, showed that indeed
the aforementioned problems, among others, could be solved by computing
predecessors of suitable regular languages and provided an $O(ps^4)$ time
algorithm for the task, where $p$ is the size of the productions and $s$
the number of states in an automaton for the regular language. Later on
(in collaboration with other authors~\cite{manyauthors}), they improved
the space
and time complexity to $O(ps^3)$. It has been pointed out that a reduction
to the satisfiability problem for Horn clauses solves the problem within
the same time and space constraints.

In this paper, we provide a further improvement which reduces the space
requirements to $O(ps^2)$ without affecting time complexity. This improvement
makes the algorithm competitive with standard algorithms for some classical
problems, e.g.\ the membership problem. The purpose of the paper is therefore
to show how a number of different decision procedures for a number of problems
can be replaced by a uniform framework, often without sacrificing efficiency.
Since the algorithm is also relatively easy to understand, we hope that this
result will be of value for educational purposes.

The rest of the paper is structured as follows:
In section~\ref{s:notation} we introduce some notations and other
preliminaries. In sections~\ref{s:prestar} and~\ref{s:impl} we present
our algorithm. We discuss applications and give a comparison with other
algorithms in section~\ref{s:app}.

\section{Notations}
\label{s:notation}

We use the notations of~\cite{HU79} for finite automata and context-free
grammars.

Fix a context-free grammar $G=(V,T,P,S)$ for the rest of the section
where $V$ is the set of variables and $T$ the set of terminals. Let
$\Sigma=V\cup T$. The set of productions~$P$ generates reachability
relations $\To$ and $\Tostar$ between strings over~$\Sigma$ in the
following sense: If $A\to\beta$ is a production of~$P$ and $\alpha$
and $\gamma$ are arbitrary strings over $\Sigma$, then
$\alpha A\gamma\To\alpha\beta\gamma$ holds (we also say that $\alpha\beta\gamma$
is derived from $\alpha A\gamma$ by application of $A\to\beta$ or that
$\alpha A\gamma$ is a direct predecessor of $\alpha\beta\gamma$).
The relation $\Tostar$ is the reflexive and transitive closure of~$\To$.
If $\alpha\Tostar\beta$ (for $\alpha,\beta\in\Sigma^*$) then we
call $\alpha$ a predecessor of $\beta$. A string $\alpha$ is called
a sentential form of~$G$ if $S\Tostar\alpha$. For a set of strings
$L\subseteq\Sigma^*$ we denote by $\prestar(L)$ the predecessors
of elements of~$L$, i.e.
\begin{quote}
$\prestar(L)=\{\,\alpha\in\Sigma^*\mid\exists\beta\in L\colon\alpha\Tostar\beta\,\}$
\end{quote}
We can represent regular subsets of $\Sigma^*$ with finite automata.
Given an automaton $A=(Q,\Sigma,\delta,q_0,F)$, the transition relation
$\mathord\to\subseteq Q\times\Sigma^*\times Q$ is inductively defined as follows:
\begin{itemize}
\item $q\newarrow{\e}q$ for all $q\in Q$.
\item if $(q,a,q')\in\delta$, then $q\newarrow{a}q'$.
\item if $(q,a,q'')\in\delta$ and $q''\newarrow{w}q'$ for some $q''\in Q$,
  then $q\newarrow{aw}q'$.
\end{itemize}
Clearly, the language accepted by the automaton is
\begin{quote}
$L(A)=\{\,w\in\Sigma^*\mid\exists \qf\in F\colon q_0\newarrow{w}\qf\,\}$.
\end{quote}

\section{Computing $\prestar(L)$}
\label{s:prestar}

In~\cite{bookotto}, Book and Otto show that for a context-free grammar
$G=(V,T,P,S)$ and a regular language $L\subseteq\Sigma^*$ the set
$\prestar(L)$ is also regular. Moreover, they provide an algorithm
for its computation.

We start from an automaton $A=(Q,\Sigma,\delta,q_0,F)$ accepting~$L$.
We obtain an automaton $A_\prestar$ accepting $\prestar(L)$ by a saturation
procedure. This procedure adds transitions to~$A$ according to the following
saturation rule:
\begin{center}\framebox{
\parbox{8.2cm}{If $A\to\beta\in P$ and $q\newarrow{\beta}q'$ in the
current automaton, \\
add a transition $(q,A,q')$.}}
\end{center}

Notice that no new states are added, and that all added transitions are
labelled with variables. Consider an example where the grammar has productions
$A~\to~a\mid BB$ and $B~\to~AB\mid b$. We apply the algorithm to the
automaton in the left half of figure~\ref{f:example}. The algorithm
will add transitions $(q_0,A,q_1)$, $(q_1,B,q_2)$, and $(q_2,A,q_1)$
because of the productions $A\to a$ and $B\to b$. Because of $B\to AB$,
we now get $(q_0,B,q_2)$ and $(q_2,B,q_2)$ which in turn lets us
apply the production $A\to BB$ to add $(q_0,A,q_2)$, $(q_1,A,q_2)$,
and $(q_2,A,q_2)$.
The right half of the figure depicts the resulting automaton.

\begin{figure}[t]
  \centering
  \begin{tikzpicture}[
      node distance=2.00cm, auto, thick, initial text={},
      scale=0.9, transform shape
    ]
   \node[state, initial] (q0)   {$q_0$};
    \node[state, right of=q0]   (q1)  {$q_1$};
    \node[state, accepting, right of=q1]   (q2)  {$q_2$};
    
    \path[->]
    (q0) edge node {$a$} (q1)
    (q1) edge[bend left=15] node {$b$} (q2)
    (q2) edge[bend left=15] node {$a$} (q1)
    ;

    \node[state, initial, right=1.2cm of q2] (r0)   {$q_0$};
    \node[state, right of=r0]   (r1)  {$q_1$};
    \node[state, accepting, right of=r1]   (r2)  {$q_2$};
    \node[below right =0.5cm and -0.5cm of q2]{{\normalsize $\begin{array}{l} A  \rightarrow a \mid  BB\\ B  \rightarrow AB \mid b \end{array}$}};
    
    \path[->]
    (r0) edge node {$a,A$} (r1)
    (r1) edge[bend left=10] node {$b, B,A$} (r2)
    (r2) edge[bend left=10] node {$a, A$} (r1)
    (r0) edge[bend left=60] node {$B,A$} (r2)
    (r2) edge[loop right] node {$B,A$} (r2)
    ;
\end{tikzpicture}
 \caption{Example automaton before (left) and after the algorithm (right).}
\label{f:example}
\end{figure}
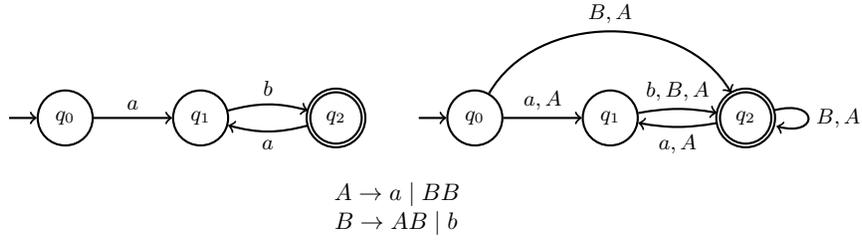

The correctness of the procedure was shown in~\cite{manyauthors}, but
for the sake of completeness we repeat the proof here.

\paragraph{Termination:}
The algorithm terminates because no new states are added and so the
number of possible transitions is finite.

\begin{lemma}
\label{l:pre1}
$L(A_\prestar)\subseteq\prestar(L(A))$
\end{lemma}

\begin{proof}
We show that the lemma holds initially and is kept invariant by the
addition rule. Let $\Newarrow{}{i}$ denote the transition relation
of the automaton after the $i$-th addition. We proceed by induction on~$i$,
i.e.\ we show that the lemma holds after $i$ additions.

\paragraph{Basis.} $i=0$. Here, $\Newarrow{}{i}$ is the transition relation
of~$A$ and clearly $L(A)\subseteq\prestar(L(A))$.

\paragraph{Step.} $i>0$. Let $\alpha$ be a word accepted by the automaton
obtained after the $i$-th step, and let $q_0\Newarrow{\alpha}{i}\qf$ be the
path by which it is accepted. Either $q_0\Newarrow{\alpha}{i-1}\qf$ holds,
too, then by the induction hypothesis $\alpha\in\prestar(L(A))$. Otherwise,
let $(q,A,q')$ be the transition added in the $i$-th step. Then $\alpha$ can
be written as $\alpha_0A\alpha_1\ldots\alpha_n$ where $n\ge1$ is the number
of occurences of the new transition, and the
path has the form $q_0\Newarrow{\alpha_0}{i-1}q\Newarrow{A}{i}q'
\Newarrow{\alpha_1}{i-1}q\ldots q'\Newarrow{\alpha_n}{i-1}\qf$. The addition
rule dictates that there is some production $A\to\beta$ and that
$q\Newarrow{\beta}{i-1}q'$. Then $\alpha'=\alpha_0\beta\alpha_1\ldots\alpha_n$
is a direct successor of $\alpha$, and $q_0\Newarrow{w'}{i-1}\qf$.
By the induction hypothesis, $\alpha'$ is in $\prestar(L(A))$ and such
is $\alpha$.
\end{proof}

\begin{lemma}
\label{l:pre2}
$\prestar(L(A))\subseteq L(A_\prestar)$
\end{lemma}

\begin{proof}
Let $\alpha$ be an element of $\prestar(L(A))$, and select $\beta\in L(A)$
such that $\alpha\Tostar\beta$. Let $i$ be the length of the derivation
$\alpha\Tostar\beta$. We proceed by induction on~$i$.

\paragraph{Basis.} $i=0$. Then $\alpha=\beta$, and $\alpha$ is accepted
by both $A$ and $A_\prestar$.

\paragraph{Step.} $i>0$. Then there exists $\gamma$ such that $\alpha\To\gamma$
and $\gamma\Tostar\beta$ in $i-1$ steps. Because of the former, there exist
$\alpha_1,\alpha_2,\alpha_3\in\Sigma^*$ and $A\in V$ such that
$\alpha=\alpha_1A\alpha_3$, \ $\gamma=\alpha_1\alpha_2\alpha_3$ and
$A\to\alpha_2\in P$. Because of the latter, by the induction hypothesis
$\gamma$ is accepted by $A_\prestar$, so there exists a path
$q_0\newarrow{\alpha_1}q\newarrow{\alpha_2}q'\newarrow{\alpha_3}\qf$ for
some $\qf\in F$ in $A_\prestar$. By the saturation rule, $(q,A,q')$ is
a transition of $A_\prestar$ and so $\alpha$ is accepted, too.
\end{proof}

\section{The Algorithm}
\label{s:impl}

We present an efficient implementation of the procedure from the previous
section. The new algorithm works on grammars which are in Chomsky normal
form extended with unit productions and $\e$-productions, i.e.\ we allow
productions of the form $A\to BC$, \ $A\to a$, \ $A\to B$, and $A\to\e$.
Notice that this is not a real
restriction since such a normal form can be obtained from an arbitrary
context-free grammar in linear time and with a linear growth in the size
of the productions. Also, assume that every element of
$\Sigma$ occurs in at least one production; otherwise, a missing element
can be removed in linear time. The algorithm is shown below.

\begin{figure}
\label{algorithm_prestar}

{\bf Input:} a context-free grammar $G=(V,T,P,S)$ in extended CNF;

\qquad\quad$\>$ an automaton $A = (Q,\Sigma,\delta,q_0,F)$\\
{\bf Output:} the automaton $A_\prestar$

\smallskip

{\algo
$\delta'\gets\emptyset;\ \eta\gets\delta;$		\label{l_deltaadd}
\forall{A\to a\in P, \ (q,a,q')\in\delta} \addto\eta{(q,A,q')}; %
							\label{l_singadd}
\forall{A\to \e\in P, \ q\in Q} \addto\eta{(q,A,q)};	\label{l_epsadd}
\while{\eta\ne\emptyset}
>remove $t=(q,B,q')$ from $\eta$;                 	\label{l_remove}
>\ifthen{t\notin\delta'}				\label{l_ifcheck}
>>\addto{\delta'}t;
>>\forall{A\to B\in P}					\label{l_forsing}
>>>\addto\eta{(q,A,q')};				\label{l_sing}
>>\forall{A\to BC\in P}					\label{l_for1}
>>>\forall{q''\ \hbox{s.t.}\ (q',C,q'')\in\delta'}	\label{l_qfor1}
>>>>\addto\eta{(q,A,q'')};				\label{l_dblfront}
>>\forall{A\to CB\in P}					\label{l_for2}
>>>\forall{q''\ \hbox{s.t.}\ (q'',C,q)\in\delta'}	\label{l_qfor2}
>>>>\addto\eta{(q'',A,q')};				\label{l_dblback}

[return] $(Q,\Sigma,\delta',q_0,F)$}
\end{figure}

Transitions which are known to belong to $A_\prestar$ are collected in
the sets $\delta'$ and~$\eta$. The latter contains transitions which still
need to be examined. No transition is examined more than once.

First, we observe that productions of the form $A\to a$ can be dealt with
once at the beginning since the algorithm will not add transitions labelled
by terminals. Similarly, the $\e$-productions are only regarded during
initialisation. It remains to deal with productions of the form
$A\to B$ and $A\to BC$
which the algorithm does in a fairly straightforward way.

We prove correctness of the new algorithm by showing that it is equivalent
to the algorithm from section~\ref{s:prestar}.

\paragraph{Termination:}
Initially, $\eta$ contains finitely many transitions. Every iteration of
the while-loop removes one of them. An iteration can add transitions only
if the check in line~\ref{l_ifcheck} is positive which is only finitely
often the case (since the size of $\delta'$ is bounded). Therefore $\eta$
will eventually become empty, and the algorithm terminates.

It remains to show that upon termination $\delta'$ is $\delta_\prestar$,
i.e.\ the smallest set containing~$\delta$ and closed under the saturation
rule. Since every element of $\eta$ eventually goes into $\delta'$, we examine
additions to $\eta$. These include the additions of $\delta$ in
line~\ref{l_deltaadd}, and it is easy to see that all other additions
fulfill the saturation rule. So $\delta'\subseteq\delta_\prestar$ holds.
For the other direction, $\delta'\supseteq\delta_\prestar$, we already
know that $\delta\subseteq\delta'$. Every other transition of $\delta_\prestar$
is due to the saturation rule. Assume that $A\to\beta\in P$ and
$q\snewarrow{\beta}q'$ in $A_\prestar$.
\begin{itemize}
\item if $\beta=\e$, then $q=q'$ and $(q,A,q')$ is added in line~\ref{l_epsadd}.
\item if $\beta=a$, then $(q,a,q')$ must be a transition of $\delta$
  (as already observed, the algorithm only adds transitions labelled with
  non-terminals). Then $(q,A,q')$ is added in line~\ref{l_singadd}.
\item if $\beta=B$, then $t_1=(q,B,q')$ must be a transition of $A_\prestar$.
  When $t_1$ is transferred from $\eta$ into $\delta'$, $(q,A,q')$ is added
  in line~\ref{l_sing}.
\item if $\beta=BC$, let $q''$ be a state such that $t_1=(q,B,q'')$ and
  $t_2=(q'',C,q')$ are transitions in $A_\prestar$. Depending on the order
  in which $t_1$ and $t_2$ occur in $\eta$, $(q,A,q')$ is added in
  line~\ref{l_dblfront} or~\ref{l_dblback}.
\end{itemize}

In~\cite{manyauthors} it was observed that a straightforward implementation
of Book and Otto's procedure may lead to an $O(p^2s^6)$ time algorithm if
$p$ is the number of productions and $s$ the number of states in~$A$.
In~\cite{ER97} an $O(ps^4)$ time algorithm was proposed,
and in~\cite{manyauthors}
the bounds were improved to $O(ps^3)$ time and space. The algorithm above
takes the same amount of time but only $O(ps^2)$ space. Moreover, we will
see later that its time complexity becomes $O(ps^2)$ for unambiguous grammars.

Imagine that $\delta'$ is implemented by a bit-array with one bit for each
element of $Q\times\Sigma\times Q$. Since $|\Sigma|=O(p)$ this takes
$O(ps^2)$ space. Membership test and addition can then be performed in
constant time. Let~$\eta$ be implemented as a stack such that finding an
element, addition and removal take constant time. Moreover, if we use a
second bit-array that keeps track of all additions to~$\eta$ ever,
we can prevent duplicate entries in~$\eta$ while still keeping the operations
constant in time. The space needed for these structures is $O(ps^2)$, too.

Line~\ref{l_deltaadd} takes at most $O(ps^2)$ time (the maximum size
of~$\delta$), the
same holds for line~\ref{l_singadd} (there are at most $s^2$ many transitions
labelled with~$a$). Line~\ref{l_epsadd} takes $O(ps)$ time. To find out the
complexity of the main loop, imagine that for every variable~$A$, there exist
three sets of productions, $A_{\it chain}$, $A_{\it front}$, and $A_{\it back}$.
A production $A\to B$ would be put into $B_{\it chain}$; a production
$A\to BC$ would occur in both $B_{\it front}$ and $C_{\it back}$. These sets
can be constructed prior to starting the algorithm proper in $O(p)$ time and
space. Then the loops in lines~\ref{l_forsing}, \ref{l_for1} and~\ref{l_for2}
need to traverse only the sets $B_{\it chain}$, $B_{\it front}$ and
$B_{\it back}$, i.e.\ those rules which are relevant for the selected
transition~$t$. Since
no element is added to $\delta'$ twice, line~\ref{l_sing} is executed
at most once for every rule $A\to B$ and states $q,q'$, i.e.\ $O(ps^2)$ times.
Similarly, lines~\ref{l_dblfront}
and~\ref{l_dblback} are both executed at most once for every combination
of productions $A\to BC$ and states $q,q',q''$, i.e.\ $O(ps^3)$ times.
Line~\ref{l_remove} is executed at most $O(ps^2)$ times.

From these observations it follows that the algorithm takes $O(ps^3)$ time
and $O(ps^2)$ space. As an aside, notice that in lines~\ref{l_qfor1}
and~\ref{l_qfor2} it is not necessary to iterate $q''$ over all elements
of~$Q$. Instead, one could maintain a set ${\it Back}(q',C)$ for every pair
$q'\in Q$ and $C\in V$ such that ${\it Back}(q',C)\ni q''$ exactly if
$(q',C,q'')\in\delta'$. Line~\ref{l_qfor1} then only needs to go through
${\it Back}(q',C)$. These sets can be updated in constant time whenever
a transition is added to~$\delta'$. Since the combined size of all ${\it Back}$
sets is no larger than $|\delta'|$, neither space nor time complexity are
affected. Similarly, one can maintain sets ${\it Front}(C,q)$ for
line~\ref{l_qfor2}.

\section{Applications}
\label{s:app}

In this section we show how several standard problems for context-free
grammars can be solved using the $\prestar$ algorithm. Let $G=(V,T,P,S)$
be a context-free grammar and let $p=|P|$.

\paragraph{Membership:}
Given a string $w\in T^*$ of length $n$, is $w\in L(G)$?
To solve the question check whether $S\Tostar w$ holds, i.e.\ if
$S\in\prestar(\{w\})$. An automaton
accepting $\{w\}$ has $n+1$ states. Hence, for a fixed grammar, the
algorithm takes $O(n^3)$ time and $O(n^2)$ space which is also
the complexity of the well-known CYK algorithm for the same problem.

However, there is more. Earley's algorithm~\cite{earley70} also takes cubic
time in general, but only quadratic time for unambiguous grammars.
Our algorithm has the same property. Assume that $G$ has no unreachable
variables. Then unambiguity implies for any pair of derivations
$A\to BC\Tostar u_1u_2$ and $A\to DE\Tostar u_3u_4$ with $u=u_1u_2=u_3u_4$
that $BC=DE$, $u_1=u_3$ and $u_2=u_4$ (otherwise there would be two
different parse trees for a word containing $u$). If $w$ consists of
terminals $w_1\ldots w_n$, an automaton accepting $\{w\}$ has transitions
of the form $(q_{i-1},w_i,q_i)$ for $1\le i\le n$. When computing
$\prestar(\{w\})$, a transition $(q_i,A,q_j)$ implies that
$A\Tostar w_{i+1}\ldots w_j$. Recall from the analysis of the algorithm
that the time complexity is dominated by the number of times
lines~\ref{l_dblfront} and~\ref{l_dblback} are executed, i.e.\ once
for every production $A\to BC$ and states $q,q',q''$ such that
$(q,B,q'')$ and $(q'',C,q')$ are transitions in the automaton.
In the unambiguous case we can deduce that once $q$ and $q'$ are known,
$q''$ is fixed. Hence, the lines are executed only $O(ps^2)$ times, and
the algorithm takes quadratic time in the number of states.

The $\prestar$ algorithm can be modified to solve the {\em parsing}
problem: Given a word $w\in L(G)$, produce a derivation $S\Tostar w$.
The modification consists of adding extra information to each transition
produced by the saturation rule. Every time a transition is added, annotate
it with the `reason' for its addition, i.e.\ with the production and the states
by which the saturation rule was fulfilled (notice that this information has
constant length for each transition, so the space complexity is not affected
by this change). After computing $\prestar(\{w\})$, take the transition
$(q_0,S,q_n)$ (if $n=|w|$) and repeatedly exploit the new information
until $w$ is derived.

\paragraph{Identifying useless variables:}
In order to rid a grammar of redundant symbols one computes the set of
useless variables. A variable $A$ is called useful if it occurs in the
parse tree of some word $w\in L(G)$, i.e.\ if there is a derivation
$S\Tostar w_1Aw_2\Tostar w$ for some $w_1,w_2\in T^*$. Otherwise $A$
is useless. Checking if a variable is useful amounts to two tests:
checking if $A$ is productive (whether there is some $w'\in T^*$ such that
$A\Tostar w'$) and checking if $A$ is reachable, i.e.\ there exist
$w_1,w_2\in T^*$ such that $S\Tostar w_1Aw_2$.
Productive variables can be identified by
computing $\prestar(T^*)$. An automaton accepting $T^*$ has one single
state, therefore this takes time $O(p)$. Notice that the test yields the set
of all productive variables. This time can also be achieved by a careful
implementation of the procedure given in~\cite{HU79}, Lemma 4.1.

To see whether a variable $A$ is reachable, we check if
$S\in\prestar(T^*AT^*)$. An automaton for this set has two states, so
this check takes linear time in the size of the grammar, too.
To find the set of all reachable variables, however, we need to repeat
the check for each variable, so the complexity would become $O(p\,|V|)$.
The procedure in~\cite{HU79} finds the set of all reachable variables
in $O(p)$, so the $\prestar$ algorithms performs worse in this case.

\paragraph{Emptiness:}
To check if the language generated by $G$ is empty, test if $S$ is unproductive,
i.e.\ $S\notin\prestar(T^*)$. We could also reduce emptiness to the
{\em containment} problem: Given a regular language~$L$, is $L(G)\subseteq L$?
Since $L(G)\subseteq L$ is equivalent to $L(G)\cap \bar L=\emptyset$ it is
sufficient to check if $S\notin\prestar(\bar L)$. If $\bar L$ can be
represented by an automaton with $s$~states, the procedure takes
$O(ps^3)$ time and $O(ps^2)$ space. For the emptiness problem we take
$L=\emptyset$; again, since $\bar L=T^*$ can be represented with one state,
the check takes $O(p)$ time and space.

\paragraph{Finiteness:}
To decide if the language generated by $G$ is finite, we first rid $G$ of all
useless symbols and recall Theorem~6.6 from~\cite{HU79}:
$L(G)$ is infinite exactly if there is a variable $A\in V$ and strings
$u,v\in T^*$ such that $A\Tostar uAv$ and $uv\ne\e$. Therefore, a solution
consists of checking whether
$A\in\prestar(T^+AT^*\mathrel\cup T^*AT^+)$ holds for some variable~$A$.
The whole procedure can be carried
out in quadratic time and linear space in the size of the grammar.
The algorithm given in~\cite{HU79} is quadratic in both time and space
(since it requires conversion to CNF) but could probably be converted
into a linear algorithm with a careful implementation.

\paragraph{Nullable variables:}
In~\cite{HU79} the problem of finding nullable variables is discussed
in the context of eliminating $\e$-productions. A variable~$A$ is called
nullable if $A\Tostar\e$ which is equivalent to $A\in\prestar(\{\e\})$.
This condition can be checked in $O(p)$ time using our algorithm.
In fact, computing $\prestar(\{\e\})$ yields the set of all nullable
variables. The procedure given in~\cite{HU79} can also be carried
out in linear time.

\section{Conclusions}

The material presented in this paper is mostly of pedagogical value.
The concept of computing $\prestar$ is relatively easy to understand
and provides a unified view of several standard problems. We think
that it would be suitable for undergraduate courses on formal
languages and automata theory where it can be used to replace the
various independent algorithms usually taught for these problems.

As a decision procedure, the algorithm is easily comprehensible;
its complexity analysis is more subtle but shows that the algorithm
is equally efficient in comparison to the standard solutions in many cases.
Of particular interest is the application to the membership problem
where the algorithm is more flexible than CYK.

Similar arguments were put forward in~\cite{manyauthors}. We have
built upon the work presented there and improved it by providing
an algorithm with better space complexity and by a more detailed
inspection of its applications. Moreover, we think that our algorithm
is easier to understand than the one in~\cite{manyauthors} which
contributes to its pedagogical merits.

\bibliographystyle{abbrv}
\bibliography{db}

\end{document}